\documentclass[preprint,12pt]{elsarticle}

\usepackage{amssymb}

\journal{arXiv}

\usepackage{amsmath,amsthm,amsfonts,amssymb,amscd}
\usepackage{graphicx}
\usepackage{rotating}
\usepackage[utf8]{inputenc}
\usepackage{dsfont}
\usepackage{comment}
\usepackage{mathrsfs}
\def\T{{\mathds{T}}}
\def\C{{\mathscr{C}}}
\def\.{{\cdot}}
\def\*{{\ast}}

\def\a{{(a_i)}}
\def\b{{(b_i)}}

\def\Null{{\textit{null}}}

\def\p4{{$P_4$}}

\def\equi {{\sim}}
\def\nequi{{\not \sim}}

\newcommand{\piso}[1]{\lfloor #1 \rfloor}

\usepackage[Algorithm]{algorithm}
\PassOptionsToPackage{noend}{algpseudocode}
\usepackage{algpseudocode}

\errorcontextlines\maxdimen

\makeatletter
\newcommand*{\algrule}[1][\algorithmicindent]{\makebox[#1][l]{\hspace*{.5em}\vrule height .75\baselineskip depth .25\baselineskip}}%
\newcount\ALG@printindent@tempcnta
\def\ALG@printindent{%
    \ifnum \theALG@nested>0
        \ifx\ALG@text\ALG@x@notext
            \addvspace{-3pt}
        \else
            \unskip
            \ALG@printindent@tempcnta=1
            \loop
                \algrule[\csname ALG@ind@\the\ALG@printindent@tempcnta\endcsname]%
                \advance \ALG@printindent@tempcnta 1
            \ifnum \ALG@printindent@tempcnta<\numexpr\theALG@nested+1\relax
            \repeat
        \fi
    \fi
    }%
\usepackage{etoolbox}
\patchcmd{\ALG@doentity}{\noindent\hskip\ALG@tlm}{\ALG@printindent}{}{\errmessage{failed to patch}}
\makeatother
\algrenewtext{EndWhile}{\algorithmicend\ \algorithmicwhile}
\algrenewtext{EndFor}{\algorithmicend\ \algorithmicfor}
\algrenewtext{EndIf}{\algorithmicend\ \algorithmicif}
\algrenewtext{EndFunction}{\algorithmicend\ \algorithmicfunction}
\algrenewtext{EndProcedure}{\algorithmicend\ \algorithmicfunction}
\algnewcommand\algorithmicto{\textbf{to}}
\algrenewtext{For}[3]%
{\algorithmicfor\ #1 $\gets$ #2 \algorithmicto\ #3 \algorithmicdo}
\algdef{SE}[DOWHILE]{Do}{doWhile}{\algorithmicdo}[1]{\algorithmicwhile\ #1}%

\newtheorem{theorem}{Theorem}
\newtheorem{lem}[theorem]{Lemma}
\newtheorem{propo}[theorem]{Proposition}
\newtheorem{coro}[theorem]{Corollary}
\newtheorem{defi}{Definition}
\newtheorem{ex}{Example}
\newtheorem{fact}{Fact}

\begin{document}

\begin{frontmatter}



\title{Cograph generation with linear delay}

\author{Átila A. Jones}
\address{Instituto de Computação, Universidade Federal Fluminense. Niterói, Brazil}
\ead{atilajones@id.uff.br}

\author{Fábio Protti}
\address{Instituto de Computação, Universidade Federal Fluminense. Niterói, Brazil}
\ead{fabio@ic.uff.br }

\author{Renata R. Del-Vecchio}
\address{Instituto de Matemática, Universidade Federal Fluminense. Niterói, Brazil}
\ead{renata@vm.uff.br}

\begin{abstract}
Cographs have always been a research target in areas such as coloring, graph decomposition, and  spectral theory. In this work, we present an algorithm to generate all unlabeled cographs with $n$ vertices, based on the generation of cotrees. The delay of our algorithm (time spent between two consecutive outputs) is $O(n)$. The time needed to generate the first output is also $O(n)$, which gives an overall $O(n\,M_n)$ time complexity, where $M_n$ is the number of unlabeled cographs with $n$ vertices. The algorithm avoids the generation of duplicates (isomorphic outputs) and produces, as a by-product, a linear ordering of unlabeled cographs wih $n$ vertices.
\end{abstract}


\end{frontmatter}



\section{Introduction}

Cographs have been defined independently by several authors since the 1970's and are usually defined as $P_4$-free graphs, as proved in \citep{corneil1981complement}. The original definition of cograph is based on a recursive construction described as follows: any single vertex graph is a cograph; if $G$ is a cograph, so is its complement graph $\overline{G}$; if $G$ and $H$ are cographs, so is their disjoint union. The \textit{disjoint union} $G=G_1\cup\,G_2$ of graphs $G_1$ and $G_2$ is a graph operation such that $V(G)=V_1 \cup V_2$ and $E(G)=E_1\cup E_2$.

Cographs can be defined alternatively in the following way. The \textit{join} $G=G_1+G_2$ of $G_1$ and $G_2$ is an operation such that $V(G)=V_1 \cup V_2$ and $E(G)=E_1\cup E_2 \cup \{xy\mid x\in E_1 \hbox{ and } y \in E_2\}$. Note that $\overline{G_1\cup G_2} =\overline{G_1}+\overline{G_2}$. One can obtain a structural decomposition of a cograph by means of the operations \textit{join} and \textit{disjoint union}, as follows. In \citep{corneil1984cographs} a special tree, called \textit{cotree}, is used to represent a cograph $G$, in which leaf nodes are associated with the vertices of $G$, and each internal node is labeled 0 ({\it type-0 node}) or 1 ({\it type-1 node}) indicating operations of join or disjoint union on their children, respectively. Furthermore, every cotree must be such that the nodes in a root-leaf path have alternate labels, which ensures that each cograph is associated with only one cotree (up to permutation of siblings). On the other hand, each cotree refers to a single cograph.

As cographs can be identified by cotrees, given a rooted tree $T$ we use the following notation. The root of $T$ is denoted by $root(T)$. Given a node $v$ (other than the root), we denote by $P_v$ the only path from $v$ to the root. The nodes of $P_v$ (except $v$) are named \textit{ancestors} of $v$ in $T$. The immediate ancestor of $v$ in $P_v$ is called the \textit{parent} of $v$ and denoted by $parent_T(v)$. We set $parent(root(T))=\Null$. If $v$ is an internal node, its \textit{children} are the elements of the set $children_T(v)=\{u \in V(T)\mid parent_T(u)= v\}$. If $v\neq root(T)$, its \textit{siblinghood} is the set $ S_T(v) = children_T(parent_T(v))$. Define also $S_T(root(T)) = \{root(T)\}$. We denote by $T(v)$ the subtree of $T$ rooted at $v$ containing all the nodes $w$ for which $v$ is an ancestor of $w$. When there is no ambiguity we omit the index $T$ from the above definitions. If two graphs $G_1$ and $G_2$ are isomorphic we write $G_1\equiv G_2$.

In \citep{corneil1981complement} the authors prove that given two leaves $v$ and $w$ of a cotree, the closest ancestor to both (that is common to $P_v$ and $P_w$) is type-1 if and only if $v$ and $w$ are adjacent in the associated cograph. Hence a cograph is connected if and only if the root of the cotree is type-1.

In the context of enumerative combinatorics, the work \citep{ravelomanana2001asymptotic} gives an asymptotic approximation for the number of cographs with a given number of vertices. Another approach consists of the generation of {\em all} graphs with a given property (e.g., ``cographs with $n$ vertices''). Observe that this kind of approach provides, as a by-product, the exact number of graphs satisfying the given property. Formally, a {\em generation} (or {\em enumeration}) algorithm $A$ generates the sequence $G_1,G_2,\ldots,G_M$ of all graphs satisfying a required property $\pi$, where $G_i\not\equiv G_j$ for $1 \leq i < j \leq M$. There are some alternative definitions of efficiency for generation algorithms~\cite{enum}, the weakest one being {\em polynomial total time}, i.e., the total running time is polynomial in the combined size of $G_1,\ldots,G_M$. The strongest definition is {\em polynomial delay}, which means that the time between the generation of two consecutive elements is polynomial only in the size of the next output element. Some important examples of enumeration algorithms with polynomial delay are the generation of all minimum spanning trees~\cite{epps}, all maximal independent sets~\cite{indsets}, and all cycles of a graph~\cite{jayme}. In~\cite{courcelle}, the author describe sufficient conditions that guarantee linear delay generation of certain structures of a graph $G$. More precisely, the author proves the following theorem: if $\mathscr{C}$ is a family of graphs of bounded tree-width then for every monadic second-order formula $\varphi(X_1,\ldots,X_p)$, where $X_1,\ldots,X_p$ are the free set variables in $\varphi$, there is an algorithm that takes as input an $n$-vertex graph $G$ in $\mathscr{C}$ and generates with linear delay all the $p$-tuples of sets that satisfy $\varphi$ in $G$, after an $O(n \log n)$ preprocessing step. The formula $\varphi$ can be interpreted as a query on $G$ with parameters $X_1,\ldots,X_p$. For instance, the result of the query $\varphi=\varphi_1\wedge\varphi_2\wedge\varphi_3$, where $\varphi_1= \forall u \, \forall v \, (X_1(u)\wedge X_1(v) \rightarrow \neg\,adj(u,v))$, $\varphi_2=\forall u \, \forall v \, (X_2(u)\wedge X_2(v) \rightarrow \neg\,adj(u,v))$, and $\varphi_3= \neg\,\exists u\,(X_1(u)\wedge X_2(u))$, is the collection of all pairs $(X_1,X_2)$ such that $X_1$ and $X_2$ are disjoint independent sets of $G$. (In other words, the result of the query can be viewed as the collection of all induced bipartite subgraphs of $G$.) 

Motivated by the above discussion, in this work we develop an algorithm that generates all (mutually nonisomorphic) unlabeled cographs with $n$ vertices. Say that two cotrees are {\em isomorphic} if they are associated with the same cograph. Since a cotree is a compact representation of the associated cograph\footnote{Note that a cotree provides an $O(n)$ size representation of a cograph, because adjacency relations are implicit from the structure of cotrees.}, our algorithm in fact generates all mutually non-isomorphic cotrees with $n$ leaves, for a given positive integer $n$. The delay of our algorithm (time spent between two consecutive outputs) is $O(n)$. The time needed to generate the first output is also $O(n)$, which gives an overall $O(n\,M_n)$ time complexity, where $M_n$ is the number of unlabeled cographs with $n$ vertices. The algorithm avoids the generation of duplicates (isomorphic outputs) and produces, as a by-product, a linear ordering of unlabeled cographs with $n$ vertices. To the best of the authors' knowledge, this is the first explicit description of a linear-delay cograph generation algorithm. Since each output element is associated with a cograph with exactly $n$ vertices, the $O(n)$ delay attained by our algorithm is best possible in practice.

The remainder of this work is organized as follows. In Sections 2 and 3, we define and study special orderings of nodes and trees using concepts of integer partitions. In Section 4, we describe our cograph generation algorithm and discuss its complexity and correctness. Section 5 contains our conclusions.

\section{Orderings of Nodes and Trees} \label{sec:node-ordering}

As discussed above, any path in a cotree alternates the types of its nodes; thus they are only determined by the type of the root. Let $\mathds{T}$ be the set of rooted trees where each internal node has at least two children. Then each element of $\T$ refers to exactly two distinct cographs: one with a type-1 root and another with a type-0 root, unless the tree consists only of the root (and, in this case, represents the trivial graph). On the other hand, each cograph is associated with a single tree of $\T$, up to permutation of nodes in the same siblinghood. In the next paragraphs we introduce a standard way to configure such a tree.

For $T\in\T$ and each node $v$ of $T$ we define $l(v)$ as the number of leaves of $T(v)$. If $v$ is a leaf, $l(v)=1$. A tree is said \textit{labeled} if each node $v$ in $T$ is labeled $l(v)$. Figure~\ref{fig:ordered-tree} shows an example of labeled tree.

\begin{sidewaysfigure}[htbp]
\centering
\includegraphics[scale=0.98]{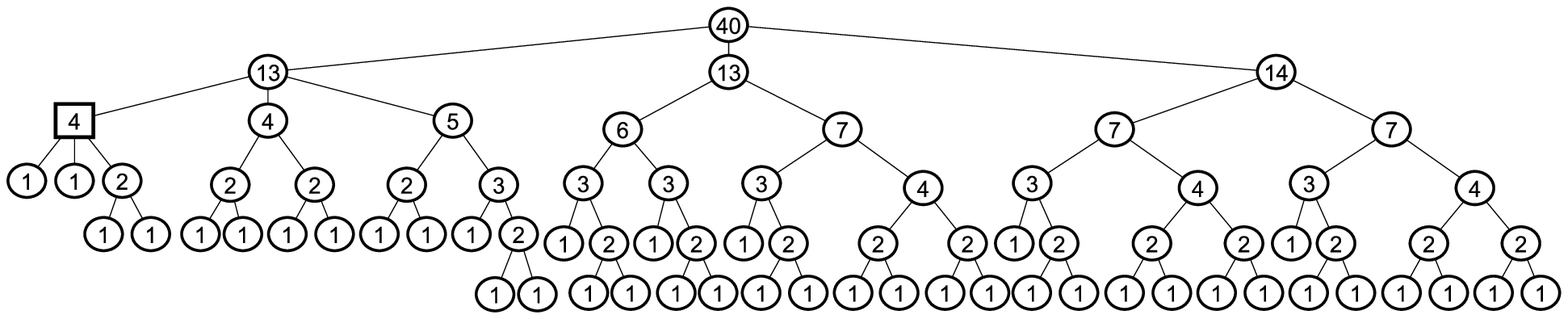}
\caption{Labeled tree.\label{fig:ordered-tree}}
\end{sidewaysfigure}

\begin{fact} \label{fact:sum-of-children}
If $v$ is an internal node then $l(v)=\sum_{w\in\mathit{children}(v)} l(w)$.
\end{fact}

Let $\T_n=\{T \in\mathds{T}\mid l(root(T))=n\}$, and consider the partition $\T=\bigcup_{n=1}^{\infty} \T_n$. Notice that each cograph with $n$ vertices can be represented by a tree in $\T_n$. The fact below is straightforward and will be useful to calculate the complexities of the algorithms.

\begin{fact} \label{fact:number-of-internal-nodes}
For each $n$ the number of nodes of a tree in $\T_n$ is at most $2n-1$.
\end{fact}

Facts~\ref{fact:sum-of-children} and~\ref{fact:number-of-internal-nodes} guarantee that labeling a tree can be done in $O(n)$.

The notion of \textit{integer partition} was introduced by Euler. In this work we assume that a partition must have at least two components, as described below.

\begin{defi}\label{def:particao}
Let $n\geq2$ be a positive integer. The {\em integer partition} of $n$ is a non-decreasing sequence of positive integers $\a_k:=(a_1, a_2, a_3, \ldots, a_k)$ such that $ \sum_ {i = 1}^{k}a_i=n$ and $ k\geq2$. The set of all partitions of $n$ is denoted by $Part(n)$.
\end{defi}

We will use a \textit{lexicographic ordering} of partitions.

\begin{defi} \label{def:ordenacaoParticao}
Given distinct elements $a, b \in Part(n)$, where $a=\a_k$ and $b =\b_m$, let $j=\min\{i\mid a_i \neq b_i\}$, where $1\leq i\leq\min\{k,m\}$. If $a_j<b_j$ then $a<b$, otherwise $b<a$.
\end{defi}

As Definition \ref{def:ordenacaoParticao} is based on comparison of integers, the law of trichotomy applies to partitions:

\begin{fact} \label{fact:trichotomy-partitions}
If $a,b \in Part(n)$ then exactly one of the following holds: $a<b$, $a=b$, or $b<a$.
\end{fact}

\begin{ex}
The elements of $Part(5)$ listed in increasing order are:
$$(1,1,1,1,1), (1,1,1,2), (1,1,3), (1,2,2), (1,4), (2,3).$$
\end{ex}

\begin{fact} \label{fact:min-and-max-partitions}
The minimum element of $Part(n)$ is $a=(1,\ldots,1)$, and the maximum is $b=\left(\lfloor\frac{n}{2}\rfloor,\lceil\frac{n}{2}\rceil \right)$.
\end{fact}

It is useful for our purposes to relate partitions of integers with distribution of siblinghoods in a tree.

\begin{defi}
Let $T\in\T$, and assume that for every internal node $v$ of $T$ its children $v_1,\ldots,v_k$ are such that $l(v_1)\leq l(v_2)\leq\cdots\leq l(v_k)$. We say that $(l(v_1),\ldots,l(v_k))\in Part(l(v))$ is the {\em partition induced by} $v$ {\em in} $T$.
\end{defi}

\begin{defi} \label{def:nodes-inequality}
Let $T \in \T$, and let $v,w$ two nodes of $T$. Define an order relation on the nodes of $T$ recursively as follows:
\begin{itemize}
\setlength\itemsep{-1.0pt}
\item[] Case 1: If $l(v)<l(w)$ then $v<w$;
\item[] Case 2: If $l(v)=l(w)=1$ then $v\equi w$;
\item[] Case 3: If $l(v)=l(w)$, $l(v)\neq 1:$
\begin{itemize}
\item[] Let $children(v)=\{v_1,\ldots,v_k\}$, and let $\a_k$ be the partition induced by $v$ in $T$. Analogously, let $children(w)=\{w_1,\ldots,w_m\}$, and let $\b_m$ be the partition induced by $w$ in $T$. Consider the following subcases:
\item[] Case 3.1: If $\a_k < \b_m$ then $v<w$;
\item[] Case 3.2: If $\a_k = \b_m$ $($and, consequently, $k=m$ $):$
\begin{itemize}
\item[] Case 3.2.a: If $v_i\equi w_i$ for all $1\leq i \leq k$ then $v\equi w$;
\item[] Case 3.2.b: Otherwise, let $j=min\{\,i\mid v_i\nequi w_i\,\}$. If $v_j<w_j$\\
\hspace*{2cm} then $v<w$ else $w<v$.
\end{itemize}
\end{itemize}
\end{itemize}
\end{defi}

If $v\equi w$ we will say that $v$ and $w$ are {\em equivalent}, which allows us to write the equivalence class $[v]=\{u \in V(T)\mid u\equi v\}$. Also, we use the following notation: (a) $v\leq w$ if and only if $v\equi w$ or $v<w$; (b) $v>w$ if and only if $v\not\leq w$.


It is worth noting that Definition~\ref{def:nodes-inequality} can be easily extended to nodes $v$ and $w$ in different trees $T_1$ and $T_2$ (simply consider an auxiliary tree with a new root $r$ with subtrees $T_1(v)$ and $T_2(w)$).

From Fact~\ref{fact:trichotomy-partitions}, we have:

\begin{fact} \em{(Node Trichotomy)} \label{fact:trichotomy-nodes}
If $v,w$ are nodes of a tree $T\in\T$ then exactly one of the following holds: $v<w$, $v\equi w$, or $v>w$.
\end{fact}

Throughout the text, if we use the symbol of equality between nodes we are referring to the same node.

For two trees $T_1,T_2\in\T$, write $T_1\equiv T_2$ if they can be converted into identical trees by sibling permutations. Definition \ref{def:nodes-inequality} gives the intuitive idea that two nodes $v$ and $w$ are equivalent when $T(v)\equiv T(w)$. This is guaranteed by the following lemma:

\begin{lem} \label{lem:equality-isomorphism}
If $v,w$ are nodes of $T\in\T_n$ then $T(v)\equiv T(w)$ if and only if $v\equi w$.
\end{lem}

\begin{proof}
The proof is trivial when $v$ is a leaf. Suppose $v$ is an internal node of $T$ and let $\a_k$ be the partition induced by $v$ in $T$, where $children(v)=\{v_1,\ldots,v_k\}$ and $a_i=l(v_i)$ for $i\in \{1,\ldots,k\}$. Analogously, let $\b_m$ be the partition induced by $w$ in $T$, where $children(w)=\{w_1,\ldots,w_m\}$.

Suppose $T(v)\equiv T(w)$. Then $\a_k=\b_k$ and $\sum a_i =\sum b_i=l(v)=l(w)=p$, for some integer $p\geq 2$. We use induction to prove that $v\equi w$ holds for every $p\geq 2$.

The base case $p=2$ easily follows, since $Part(2)=\{(1,1)\}$. Now assume that $v\equi w$ holds for every $q<p$. Since $T(v)\equiv T(w)$, we can arrange siblings so that $T(v_i)\equiv T(w_i)$ and $l(v_i)=l(w_i)<p$ for $i\in\{1,\ldots,k\}$. Thus, by the induction hypothesis, $v_i\equi w_i$ holds, which implies $v\equi w$ by Case 3.2.a in Definition~\ref{def:nodes-inequality}.

Conversely, suppose that $v\equi w$. Then $l(v)=l(w)=p$ for some $p$. We use induction on $p$ again to prove that $T(v)\equiv T(w)$. The result is valid for the base case $p =2$. Suppose now that it is valid for every $q<p$. Using node trichotomy, we have that $v_i\equi w_i$ for $i\in\{1,\ldots,k=m\}$, with $l(v_i)=l(w_i)<p$. Thus, by the induction hypothesis, $T(v_i)\equiv T(w_i)$, i.e., $T(v)\equiv T(w)$.
\end{proof}

\begin{defi} \label{def:ordered-tree}
Let $T\in\T$, and let $v$ be a node of $T$. The siblinghood $children_T(v)=\{v_1,\ldots,v_k\}$ is said to be {\em ordered} if $v_1\leq\cdots\leq v_k$. If all the siblinghoods of a tree are ordered, we say that the tree is ordered.
\end{defi}

The tree in Figure~\ref{fig:ordered-tree} is ordered. Ordered trees establish a standard way to represent cotrees, as shown in the following proposition:

\begin{propo} \label{propo:cographs-and-trees}
Each cograph is associated with a single ordered tree in $\T$.
\end{propo}
\begin{proof}
Let $G$ be a cograph and let $T$ be its associated cotree. In~\citep{corneil1981complement} the authors prove that $T$ is unique up to isomorphism. Let $T'$ and $T''$ be ordered trees, both isomorphic to $T$. As the nodes of $T'$ and $T''$ have an order of arrangement, they are the same tree, up to permutation of {\em equivalent} nodes of the same siblinghood. But Lemma~\ref{lem:equality-isomorphism} ensures that this permutation generates identical subtrees with exactly the same arrangement of nodes. Therefore $T'$ and $T''$ are identical.
\end{proof}

Now we introduce a total ordering on trees of $\T_n$.

\begin{defi} \label{def:ordering-trees}
Let $T_1, T_2\in \T_n$. Define $T_1<T_2$ if $root(T_1)<root(T_2)$, and $T_1\equi T_2$ if $root(T_1)\equi root(T_2)$.
\end{defi}

\begin{fact} \label{fact:minimum-tree}
The minimum element of $\T_n$ is the tree $T_{\mathit{min}}$ whose root induces the partition $(1,\ldots,1)\in Part(n)$ (see Fact~\ref{fact:min-and-max-partitions}). This tree is associated with the graphs $K_n$ and $\overline{K_n}$.
\end{fact}

\section{Generating the next element}\label{sec:next-generation}

Proposition \ref{propo:cographs-and-trees} ensures that one can generate the unlabeled cographs with $n$ vertices by generating the elements of $\T_n$, whose first element is the minimum tree $T_{\mathit{min}}$ described in Fact~\ref{fact:minimum-tree}. Based on the total order established in Definition~\ref{def:ordering-trees}, a way to generate all the elements of $\T_n$ is to generate the tree $T_2$ which is {\em immediately next} to $T_{\mathit{min}}$, and then the tree $T_3$ which is immediately next to $T_2$, and so on, until the generation of the maximum tree in $\T_n$. The notion of immediately next element in a set equipped with a total order is formalized in the definition below.

\begin{defi}
Let $\mathbb{X}$ be a non-empty, finite, and totally ordered set. Given $a \in \mathbb{X}$, we say that $b$ is the element {\em immediately next} to $a$ in $\mathbb{X}$ if $b \in \mathbb{X}$, $b>a$, and for any $c\in \mathbb{X}$ such that $c>a$ it holds that $c\geq b$. If there no element immediately next to $a$ in $\mathbb{X}$ then $a$ is the maximum element of $\mathbb{X}$.
\end{defi}

We now apply the above definition to the set $Part(n)$. Given a partition $\a_k\in Part(n)$,  Algorithm~\ref{alg:immediately-next-partition} determines the partition immediately next to $\a_k$, if it exists, or decides that $\a_k$ is the maximum partition in  $Part(n)$ (see Fact~\ref{fact:min-and-max-partitions}).

\begin{algorithm}[htbp]
\caption{$\mathit{NextPartition}$}
 \label{alg:immediately-next-partition}
\begin{algorithmic}[1]
\Require $\a_k \in Part(n)$
\Ensure $\b_m$ immediately next to $\a_k$ in $Part(n)$, if it exists, or \Null
\Procedure{$\mathit{NextPartition}$}{$\a_k$}
  \State $n \gets \sum_{i=1}^k a_i$
  \If {$a_1 \neq \lfloor n/2\rfloor$}
     \If {$a_k-a_{k-1}\leq 1$}
     	\Return $(a_1,a_2,\ldots,a_{k-2},a_{k-1}+a_k)$;
     \Else \Comment $a_k-a_{k-1} > 1$
        \State $a_{k-1}\gets a_{k-1}+1$;
        \State $a_k \gets a_k -1$;
        \State $q\gets a_k \ \mathrm{div} \ a_{k-1}$; \ \ $r\gets a_k\mod a_{k-1}$;
        \State \textbf{if} $q>1$ \textbf{then} \Return $(a_1,\ldots,a_{k-2},\underbrace{a_{k-1},\ldots,a_{k-1}}_{q \hbox{ {\scriptsize times}}},(a_{k-1}+r))$;
        \State \textbf{else} \Return $(a_1,\ldots,a_{k-2},a_{k-1},a_k)$;
     \EndIf
   \Else
	 \If {$n\neq 3$} \Return \Null
	 \Else \Comment $n=3$ special case
	       \State \textbf{if} {$a_2=\lceil n/2\rceil$} \textbf{ then } \Return \Null ;
	       \State \textbf{else } \Return $(1,2)$;
	 \EndIf
   \EndIf
\EndProcedure
\end{algorithmic}
\end{algorithm}

\begin{theorem} \label{thm:next-partition-correctness}
If $a\in Part(n)$, Algorithm~\ref{alg:immediately-next-partition} returns the partition immediately next to $a$. In case the algorithm returns \Null, $a$ is the maximum element of $Part(n)$.
\end{theorem}

\begin{proof}

Let $a=\a_k \in Part(n)$. Cases $n=2$ and $n=3$ are trivial.

By definition of partition, necessarily $a_1\leq \piso{\frac{n}{2}}$. Suppose $n>3$ and $a_1<\piso{\frac{n}{2}}$. Consider the following cases:

\bigskip

\noindent {\em Case 1}: $a_k-a_{k-1} \leq 1$.

\medskip
	
In this case, the algorithm returns the partition $b:=\b_{k-1}$, where $b_i=a_i$ for $i\in \{1,\ldots,k-2\}$ and $b_{k-1}=a_{k-1}+a_k$.
	
Clearly, $b>a$ and $b_{k-1} \in Part(n)$.	
	
We prove that $b$ is the partition immediately next to $a$. Let $c=(c_i)_m\in Part(n)$ such that $c>a$. Our aim is to prove that $c\geq b$. By definition there is a minimum index $j\leq k$ such that $c_j>a_j$.

If $j\leq k-2$ then $c_j>a_j=b_j$, and it follows that $c>b$.

If $j=k-1$ then $c_i=a_i=b_i$ for $i \in \{1,\ldots,k-2\}$ and $c_{k-1}>a_{k-1}$. Thus, $c_{k-1} \geq a_{k-1}+1 \geq a_k$. Also, note that $c$ has $k-1$ terms, otherwise we would have $c_k\geq c_{k-1} \geq a_k$. This implies $\sum^k c_i>\sum^k a_i=n$, which is a contradiction. Therefore,
$$\sum_{i=1}^{k-2}c_i+c_{k-1}=\sum_{i=1}^{k-2}a_i+a_{k-1}+a_k,$$
and this implies $c_{k-1}=a_{k-1}+a_k=b_{k-1}$, that is, $c=b$.
		
Finally, if $j=k$ then $\sum^k c_i>\sum^k a_i=n$, which is a contradiction. Hence, $c\geq b$.
	
\bigskip

\noindent {\em Case 2}: $a_k-a_{k-1}>1$.

\medskip

In this case, lines 6 and 7	generate the partition $b:=\b_k=(b_1,\ldots,b_k)$, where $b_i=a_i$ for $i \in \{1,\ldots,k-2\}$, $b_{k-1}=a_{k-1}+1$, and $b_k=a_k-1$.
		
Notice that $\sum^{k} b_i=\sum^{k} a_i=n$, and then $b \in Part(n)$. Let $q$ and $r$ be the quotient and the rest, respectively, of the integer division of $b_k$ by $b_{k-1}$ (line 8). We have two sub-cases to consider:

\bigskip
	
\noindent {\em Case 2a}: $q>1$.

\medskip
		
In this case, the algorithm returns the partition
$$b'=(b_1,\ldots,b_{k-2},b_{k-1},\underbrace{b_{k-1},\ldots,b_{k-1}}_{q-1 \ \hbox{times}}, b_{k-1}+r).$$

Note that $b'$ has $k+q-1$ terms. Note also that:
\begin{itemize}
\item[] (a) $b'_i=a_i$, if $i\in \{1,\ldots,k-2\}$;
\item[] (b) $b'_i=a_{k-1}+1$, if $i\in\{k-1,\ldots,k+q-2\}$;
\item[] (c) $b'_{k+q-1}=a_{k-1}+1+r$.
\end{itemize}

\noindent Thus:
	
\begin{eqnarray*}
\sum^{k+q-1}_{i=1} b'_i
&=& (\sum^{k-2}_{i=1} b'_i)+b'_{k-1}+(\sum^{k+q-2}_{i=k} b'_i)+b'_{k+q-1}\\
&=&(\sum^{k-2}_{i=1} a_i)+(a_{k-1}+1)+(a_{k-1}+1)(q-1)+(a_{k-1}+1+r)\\
&=&\sum^{k}_{i=1} a_i=n.
\end{eqnarray*}	
	
Therefore $b'\in Part(n)$ and $b'\geq a$.
	
We prove that $b'$ is immediately next to $a$ in $Part(n)$. Let $c=(c_i)_m \in Part(n)$ such that $c>a$, and let $j\leq k$ be the minimum index such that $c_j>a_j$.
	
If $j\leq k-2$ then $c_j>a_j=b'_j$, i.e., $c>b'$.
	
If $j=k-1$ then $c_{k-1}>a_{k-1}$, i.e., $c_{k-1}\geq a_{k-1}+1=b'_{k-1}$. Thus, $c\geq b'$.
	
Finally, if $j=k$ then $\sum^k c_i>\sum^k a_i=n$, which is a contradiction. Thus $c\geq b'$.

\bigskip
	
\noindent {\em Case 2b}: $q=1$.

\medskip
	
In this case the algorithm returns the partition $b=\b_k$. Since $b_{k-1}=a_{k-1}+1>a_{k-1}$, we have $b>a$.
	
Again, we prove that $b$ is immediately next to $a$ in $Part(n)$. Let $c=(c_i)_m \in Part(n)$ such that $c>a$, and let $j\leq k$ be the minimum index such that $c_j>a_j$.
	
If $j\leq k-2$ then $c_{k-2}>a_{k-2}=b_{k-2}$, i.e., $c>b$.
	
If $j=k-1$ then $c_{k-1}>a_{k-1}$, i.e., $c_{k-1}\geq a_{k-1}+1=b_{k-1}$. Thus, $c\geq b$.
	
Finally, if $j=k$ then $\sum^k c_i > \sum^k a_i=n$, a contradiction. Thus $c\geq b$.

To conclude the proof, if $n>3$ and $a_1=\lfloor\frac{n}{2}\rfloor$ then the algorithm returns \Null.

\end{proof}

\begin{coro}\label{coro:next-partition-complexity}
Let $a\in Part(n)$. Then $\mathit{NextPartition}(a)$ runs in $O(n)$ time.
\end{coro}

\begin{proof}
The only case that does not have a constant-time complexity appears in line 9, whose worst case occurs when the input is $a=(1,n-1)\in Part(n)$. In such a situation, $q=\piso{\frac{a_k-1}{a_{k-1}+1}}=\frac{n-2}{2}$ and the number of operations performed is $O(q)=O(n)$.\end{proof}

In order to develop a procedure for determining the tree immediately next to the current tree being generated, we introduce the concept of {\em pivot node}, which indicates the place of the current tree where changes will be made.

\begin{defi}
Let $T \in \T$ be an ordered tree. A node $v$ of $T$ is said to be {\em exhausted} if it is a leaf or the partition induced by $v$ in $T$ is the maximum element of $Part(l(v))$.
\end{defi}

\begin{fact} \label{fact:exhausted-node}
Given an ordered and labeled tree $T \in \T_n$, verifying whether a node $v$ of $T$ is exhausted can be done in constant time (see Fact~\ref{fact:min-and-max-partitions}).
\end{fact}

For an ordered tree $T\in \T$ we define the {\em inverted post-order traversal} of $T$ by the following recursive procedure: if $children(root(T))=\{v_1 \leq \cdots \leq v_k\}$, traverse (in the given sequence) the subtrees $T(v_k),T(v_{k-1}),\ldots,T(v_2),T(v_1)$ in inverted post-order, and then visit $root(T)$. The inverted post-order traversal can be done in $O(n)$ time, by Fact~\ref{fact:number-of-internal-nodes}.

\begin{defi} \label{def:tree-pivot}
Let $T\in\T$ be an ordered tree such that the inverted post-order traversal of $T$ visits its nodes in the sequence $v_1,\ldots,v_m$, where $m=|V(T)|$. The {\em pivot} of $T$, if it exists, is the node $v_i$ with minimum index such that $v_i$ is not exhausted.
\end{defi}

The square node in Figure~\ref{fig:ordered-tree} is the pivot of the tree.

Given an ordered and labeled tree $T \in \T_n$, the procedure for determining the pivot of $T$ can be simply done using the definition. In other words, simply perform an inverted post-order traversal and check whether the current visited node is exhausted or not. We call this procedure $\mathit{FindPivot}(T)$. If there is no pivot in the tree, the procedure returns \Null.

\begin{fact} \label{fact:pivot-complexity}
The search for the pivot node in a tree $T\in\T_n$ can be done in $O(n)$ time.
\end{fact}

Another step towards our cograph generation algorithm is to describe a procedure that replaces a subtree $T(v)$ of a given ordered $T\in\T_n$ by another subtree $T'$ rooted at $v$ with the same number of leaves as $T(v)$. The procedure receives as input a node $v$ of $T$ and a partition $\a_k \in Part(l(v))$, and replaces $T(v)$ by a subtree $T'$ such that $root(T')=v$, $children_{T'}(v)=\{v_1,\ldots,v_k\}$, and, for $1\leq i\leq k$, $a_i=l(v_i)$ and $children_{T'}(v_i)$ consists of $a_i$ leaves (if $a_i>1$). The procedure returns the new tree $T_R$ obtained from $T$ in this way.


\begin{algorithm}[htbp]
\caption{$\mathit{RebuildNode}$}
\label{alg:reconstroiVertice}
\begin{algorithmic}[1]
\Require node $v$ of an ordered tree $T\in\T_n$; partition $\a_k \in Part(l(v))$
\Ensure new tree $T_R$ obtained from $T$ where $T(v)$ is replaced by a subtree $T'$ such that $\a_k$ is the partition induced by $v$ in $T'$
\Procedure{$\mathit{RebuildNode}$}{$v,\a_k$}
  \State Replace the children of $v$ by $v_1,\ldots,v_k$ such that $\a_k$ becomes the partition induced by $v$
  \For{$i$}{1}{$k$}
   \State \textbf{if} {$a_i>1$} \textbf{then} insert $a_i$ leaves as children of $v_i$
  \EndFor
  \State \Return $T$
\EndProcedure
\end{algorithmic}
\end{algorithm}

\bigskip

By Fact~\ref{fact:sum-of-children}, we have:

\begin{fact} \label{fact:rebuild-complexity}
Let $\a_k \in Part(l(v))$. Then procedure $\mathit{RebuildNode}(v,\a_k)$ runs in $O(l(v))$ time.
\end{fact}

\begin{lem} \label{lem:maximum-tree-no-pivot}
A tree $T\in\T_n$ is the maximum element of $\T_n$ if and only if $T$ contains no pivot node.
\end{lem}
\begin{proof}
Suppose $T$ contains a pivot node $v$ such that the partition induced by $v$ in $T$ is $a\in Part(l(v))$. Thus there is $b\in Part(l(v))$ such that $b>a$. Let $T_R$ be the tree returned by $\mathit{RebuildNode}(v,b)$. Clearly, $T_R>T $ and therefore $T$ is not the maximum element of $\T_n$. The other direction is trivial.\end{proof}

Now we are able to describe the procedure that generates the tree immediately next to a given tree $T$, which is crucial to our cograph generation algorithm.

We use the following notation: Let $S(v)=\{v_1,\ldots,v_m\}$ be the ordered siblinghood of a node $v$. Assume $v=v_j$ for some $j\in\{1,\ldots,m\}$ and define the sets $S^+(v_j)=\{v_{j+1},\ldots,v_m\}$ and $S^-(v_j)=\{v_1,\ldots,v_{j-1}\}$. Note that $S^-(v_j) \cup \{v_j\} \cup S^+(v_j)$ is a partition of $S(v)=S(v_j)$.

\begin{algorithm}[H]
\caption{$\mathit{NextTree}$}
\label{alg:immediately-next-tree}
\begin{algorithmic}[1]
\Require ordered and labeled tree $T \in \T_n$
\Ensure $T'\in \T_n$ such that $T'$ is the tree immediately next to $T$
\Procedure{$\mathit{NextTree}$}{$T$}
  \State $v \gets \mathit{FindPivot}(T)$
  \If {$v \neq \Null$}
	\State $b_m \gets \mathit{NextPartition}(\a_k)$, where $\a_k$ is the partition induced by $v$
	\State $v \gets \mathit{RebuildNode}(v,\b_m)$
	    \State $x \gets v$
	\Repeat
	  \ForAll {$y \in S^+(x)$}
	      \State \textbf{if} $l(y)=l(x)$ \textbf{then} copy subtree $T(x)$ in $T(y)$
	  \State \textbf{else} $y \gets \mathit{RebuildNode}(y,c)$, where $c=(1)_{l(y)}$
	  \EndFor
	  \State $x \gets parent(x)$
	\Until{$x = \Null$}
    \State \Return $T$
  \Else \hbox{ } \Return \Null \Comment there is no immediately next tree
  \EndIf
\EndProcedure
\end{algorithmic}
\end{algorithm}

Algorithm~\ref{alg:immediately-next-tree} finds the tree pivot $v$ and the partition that is immediately next to the one induced by $v$. Next, it ``restarts'' each node $w$ visited during $\mathit{FindPivot}(T)$ to the lowest possible configuration for $w$. This idea is based on the fact that comparison between trees is done from left to right in each siblinghood, while the searching for the pivot occurs from right to left.

Figure~\ref{fig:next-tree} depicts the tree immediately next to the tree depicted in  Figure~\ref{fig:ordered-tree}.

\begin{sidewaysfigure}[htbp]
\centering
\includegraphics[scale=0.98]{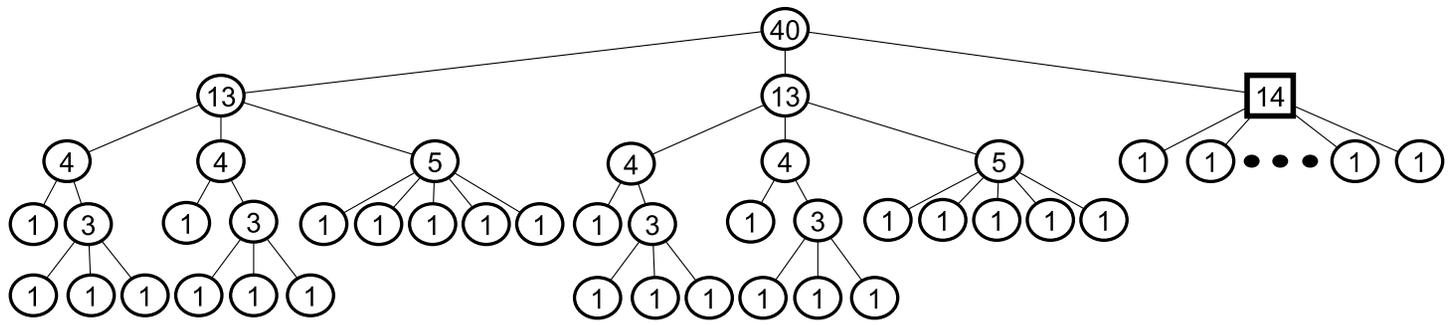}
\caption{Tree immediately next to the tree in Figure~\ref{fig:ordered-tree}. Its pivot is the square node.\label{fig:next-tree}}
\end{sidewaysfigure}

Before checking the correctness of Algorithm~\ref{alg:immediately-next-tree}, we prove the following lemma.

\begin{lem} \label{lem:next-tree-keeps-ordering}
Let $T\in\T_n$ be an ordered tree that is not the maximum element of $\T_n$. Then the tree returned by $\mathit{NextTree}(T)$ is ordered.
\end{lem}
\begin{proof}

Let $T'$ be the tree returned by $\mathit{NextTree}(T)$. Note that $T'\in \T_n$ because its number of leaves is the same as in $T$. Let $v$ denote the pivot of $T$ (whose existence is guaranteed by Lemma~\ref{lem:maximum-tree-no-pivot}). Denote its corresponding node in $T'$ by $v'$ (line 5). In general, for each $x$ in $T$, denote by $x'$ its corresponding node in $T'$, if it exists (lines 9 and 10).

For each $x'$ ancestor of $v'$, we prove that the siblinghoods $ S_{T'}(x')$ and $S_{T'}(v')$ are ordered. Since other siblinghoods did not change, their orderings follow directly from the ordering of $T$. Moreover, by the same reason, the set $S_{T'}^-(x') $ is already ordered for each $x'$. Then it remains to analyze the orderings of $S^+_{T'}(x')$ and $S^+_{T'}(v')$. In fact, for every $y' \in S_{T'}^+(x') \cup \{x'\}$ the following cases are valid:

\begin{enumerate}
\setlength\itemsep{-1.0pt}
\item If $x \leq y$ and $l(x)=l(y)$ then the operation in line 9 and Lemma  \ref{lem:equality-isomorphism} guarantee $x'\equi y'$.

\item If $x<y$ and $l(x)<l(y)$ then the partition induced by $y'$ is minimum in $Part(l(y))$ (line 10); therefore, $x'<y'$.
\end{enumerate}

Thus in both cases the ordering of the nodes is maintained, i.e. $S_{T'}(x')$ is ordered. It remains to analyze the ordering of $S_{T'}(v')$.

Write $S_{T}(v)=\{ v_1 \leq \ldots \leq v_j \leq v_{j+1} \leq \ldots \leq v_{j+h}\}$, where $v_j=v$. Similarly, write $S_{T'}(v')=\{ v'_1 \leq \ldots \leq v'_j \leq v'_{j+1} \leq \ldots \leq v'_{j+h}\}$, where $v'_j=v'$.

As the algorithm changes only the nodes of $S^+_T(v)\cup \{v\}$, the ordering of $T$ guarantees that $S_{T'}^-(v')$ is ordered.

\medskip

\noindent \textbf{Claim:} $v'$ is correctly positioned in $S_{T'}(v')$.

\medskip

To prove the Claim we need to show that $v'_{j-1} \leq v'_j \leq v'_{j+1}$. By Theorem~\ref{thm:next-partition-correctness} and Case $4.1$ of Definition~\ref{def:nodes-inequality} we have $v'_j> v_j $. But by the ordering of $T$, $v_j\geq v_{j-1}\equi v'_{j-1}$ is valid, and thus $v'_j\geq v'_{j-1}$. On the other hand, if $l(v_{j+1})=l(v_j)$ then by line 9 we have $v'_{j+1}\equi v'_j$, otherwise $l(v_{j+1})>l(v_j)$ and from line 10 and Definition~\ref{def:nodes-inequality} we have $v'_{j+1}>v'_j$. Hence, the Claim follows.

\medskip

Following the same analysis used for $S_{T'}^+(x')$, we conclude that $S_{T'}^+(v')$ is also ordered, because both are handled by the operation in lines 9 and 10. This fact and the above Claim guarantee that $S_{T'}(v')$ is ordered. Hence, the lemma follows.
\end{proof}

\begin{theorem} \label{thm:next-tree-correctness}
Let $T\in\T_n$. If $\mathit{NextTree}(T)=\Null$ then $T$ is the maximum element in $\T_n$; otherwise,  $\mathit{NextTree}(T)$ returns the tree immediately next to $T$ in $\T_n$.
\end{theorem}

\begin{proof}
If the procedure returns \Null, the result is guaranteed by Lemma~\ref{lem:maximum-tree-no-pivot}. Otherwise, let $T'$ be the output of the algorithm.

Since $T'$ is ordered, we use Lemma~\ref{lem:next-tree-keeps-ordering} to compare $T$ and $T'$. Let $v$ be the pivot of $T$, and $v'$ the node corresponding to $v$ in $T'$. As in the previous lemma, for each $x$ in $T$ we denote by $x'$ its corresponding node in $T'$.

\medskip

\noindent \textbf{Claim}: $T'>T$.

To prove the claim, first note that Theorem~\ref{thm:next-partition-correctness} and Definition~\ref{def:nodes-inequality} guarantee $v'>v$.

If $v=root(T)$ then it is clear that $T'>T$. Otherwise we write $S_{T}(v)=\{v_1\leq\ldots\leq v_k\leq\ldots\leq v_{h}\}$ and $S_{T'}(v')=\{v'_1\leq\ldots\leq v'_k\leq\ldots\leq v'_{h}\}$, where $v_k=v$ and $v'_k=v'$. Since the algorithm only makes changes to $S^+_T(v_k)$, we have $v_i\equi v'_i$ for $1 \leq i <k$. Hence, by Case 3.2.b in Definition~\ref{def:nodes-inequality}, $parent(v')>parent(v)$.

Let $x=parent(v)$. Write $S_T(x)=\{x_1\leq\ldots\leq x_j\leq\ldots\leq x_u\}$ and $S_{T'}(x')=\{x'_1\leq\ldots\leq x'_j\leq\ldots\leq x'_u\}$, where $x_j= x$ and $x'_j= x'$. Similarly, we have $x_i\equi x'_i$ for $1\leq i <j$. Then, since $x'>x$, by Case 3.2.b in Definition~\ref{def:nodes-inequality} it follows that $parent(x')>parent(x)$. In general, the same argument can be successively applied to each ancestor $x$ of $v$ to conclude that $x'>x$. If $x=root(T)$ then $T'>T$, and the claim follows.

\medskip

Now prove that $T'$ is immediately next to $T$ in $\T_n$.

Initially, the algorithm takes the pivot $v$ of $T$ and builds $v'$, whose induced partition by line 5 is immediately next to the partition induced by $v$. Let $x$ be such that $x=v$ or $x$ is an ancestor of $v$. Then, for each $y\in S_T(x)$, the algorithm builds $y'_i\in S_{T'}(x')$ as follows: if $y\in S_T^-(x)$ then $y'$ is equivalent to $x'$; if $y\in S_T^+(x)$ then $y$ is exhausted from the definition of pivot. Therefore, $y$ is built so that it is the minimum node, as detailed below and similarly to the proof of the previous lemma:

\begin{enumerate}
\setlength\itemsep{-1.0pt}
\item If $x\leq y$ and $l(x)=l(y)$ then $x'\equi y'$ (line 9).

\item If $x<y$ and $l(x)<l(y)$ then $y'$ is built from the minimum partition $Part(l(y))$; therefore, $x'<y'$ (line 10).
\end{enumerate}

Since $j=\min\{\;i \mid x_i\nequi x'_i\;\}$ and $x'_j> x_j $, it follows that $parent(x')>parent(x)$. By the cases above, for each $i\in\{j+1,\ldots,u\}$ we have that $x_i$ is exhausted in $T$. Moreover, $x'_i$ is built so that it is the minimum node with $l(x_i)$ leaves. Hence, $parent(x')$ is immediately next to $parent(x)$. As this argument applies to every ancestor $x$ of $v$ and to $v$ itself, it follows that $T'$ is immediately next to $T$ in $\T_n$.
\end{proof}

\begin{coro}\label{coro:next-tree-complexity}
Let $T\in\T_n$. Then algorithm $\mathit{NextTree}(T)$ runs in $O(n)$ time.
\end{coro}

\begin{proof}
The worst case occurs when $T$ has a pivot other than $root(T)$.

Corollary~\ref{coro:next-partition-complexity} together with Facts~\ref{fact:pivot-complexity} and~\ref{fact:rebuild-complexity} ensure that all the operations outside the loop in lines 7--13 are done in $O(n)$ time.

The {\bf for} loop in lines 8--11 performs the operations in lines 10--11, for each $y\in S^+(x)$. Assume that $S(x)=\{x_1,\ldots,x_k\}$. The worst case complexity of lines 8--11 occurs when $x=x_1$ and $S^+(x)=\{x_2,\ldots,x_k\}$, and thus is given by $\sum_{i=2}^k O(l(x_i))=O(l(parent(x))$ (recall Fact~\ref{fact:sum-of-children}).

Let $P_v: (v_0,v_1,\ldots,v_{k-1},v_k)$ be the path from the pivot to the root of the tree, where $v_0=v$ and $v_k=root(T)$. The {\bf repeat} command in lines 7--13 executes the internal {\bf for} loop in lines 8--11 for each $x$ of $P_v$. Then, using the idea developed in the previous paragraph, we conclude that: (a) for $x=v_0$, the internal {\bf for} loop in lines 8--11 is executed for every $y\in S^+(v_0)$ in $O(l(v_{1}))$ total time, since $v_1=parent(v_0)$; (b) the runtime for $x=v_i$, $i\in\{1,\ldots,k-1\}$, is:

$$O(l(v_i))+\sum_{y\in S^+(v_i)}O(l(y))) = O(l(v_{i+1})),$$

\noindent where $O(l(v_i))$ is the complexity of the $i$ previous iterations. The process terminates when the last ancestor ($x=v_k$) is reached. Therefore, the overall time complexity of $\mathit{NextTree}(T)$ is $O(l(v_k))=O(n)$.
\end{proof}

\section{Cograph generation}\label{sec:cograph-generation}

As discussed in the beginning of the previous section, the generation of the elements in $\T_n$ can be done as follows: starting from the minimum tree $T_{\mathit{min}}$ in $\T_n$  (Fact \ref{fact:minimum-tree}), apply Algorithm \ref{alg:immediately-next-tree} to generate the tree immediately next to $T_{\mathit{min}}$ in $\T_n$, and successively repeat the application of the algorithm until reaching the maximum element of $\T_n$ (characterized by the absence of pivot). Since the entire generation procedure starts with an ordered tree (recall that $T_{\mathit{min}}$ is trivially ordered), Lemma~\ref{lem:next-tree-keeps-ordering} ensures that subsequent trees are all ordered as well, i.e., there is no need of extra work to order trees along the generation. In addition, by Corollary~\ref{coro:next-tree-complexity}, each new generated tree is determined in $O(n)$ time. Based on such arguments, we present below a formal description of our cograph generation algorithm.

Let $\C$ denote the family of all cographs, and consider the partition $\C=\bigcup_{n=1}^{\infty}\C_n$, where $\C_n=\{G\in\C\mid |V(G)|=n\}$. Below we establish a total order on the members of $\C_n$. First, note that any cotree $T$ with $n$ leaves can be viewed as a member (ordered tree) of $\T_n$, and therefore we can apply the total ordering in Definition~\ref{def:ordering-trees} to cotrees.

\begin{defi} \label{def:cograph-ordering}
Let $G_1,G_2\in\C_n$, and let $T_1,T_2$ be their respective ordered cotrees. Say that $G_1>G_2$ if $(i)$ $T_1>T_2$, or $(ii)$ $T_1\equi T_2$, $root(T_1)$ is type-1, and $root(T_2)$ is type-0. In addition, say that $G_1=G_2$ if $T_1\equi T_2$ and $root(T_1)$ and $root(T_2)$ are of the same type.
\end{defi}

From the above definition, it is easy to see that $G_1>G_2$ implies $G_1\not\equiv G_2$. In adition, $G_1=G_2$ if and only if $G_1\equiv G_2$.

For $T\in\T_n$, let $T^0$ (resp., $T^1$) be the cotree associated with $T$ by setting $root(T)$ as a type-0 (resp., type-1) node. Also, let $G_i^0$ and $G_i^1$ be the cographs associated with cotrees $T_i^0$ and $T_i^1$, respectively.

\begin{algorithm}[H]
\caption{$\mathit{CographGeneration}$}
\label{alg:cograph-generation}
\begin{algorithmic}[1]
\Require integer $n\geq 2$
\Ensure all cographs with $n$ vertices
\Procedure{$\mathit{CographGeneration}$}{$n$}
  \State $T_1\gets T_{\mathit{min}}$  \Comment{the minimum tree in $\T_n$}
  \State $i\gets 1$
  \Repeat
    \State Output $T_i^0, T_i^1$
    \State $T_{i+1} \gets \mathit{NextTree}(T_i)$
    \State $i\gets i+1$
  \Until{$T_i = \Null$}
 \EndProcedure
\end{algorithmic}
\end{algorithm}

\begin{theorem} \label{propo:corretude_GeradorCografos}
Let $n \geq 2$ be an integer. Then:

\noindent $(a)$ the sequence $T_1^0,T_1^1,T_2^0,T_2^1,\ldots$ generated by Algorithm~\ref{alg:cograph-generation} contains all the cotrees with $n$ leaves;

\noindent $(b)$ the associate sequence $S=G_1^0,G_1^1,G_2^0,G_2^1,\ldots$ contains all the cographs with $n$ vertices, where no two graphs in $S$ are isomorphic;

\noindent $(c)$ the delay of Algorithm~\ref{alg:cograph-generation} is $O(n)$;

\noindent $(d)$ the time spent by Algorithm~\ref{alg:cograph-generation} to output the first element is $O(n)$.

\end{theorem}

\begin{proof}

\noindent (a) By Theorem~\ref{thm:next-tree-correctness}, the tree $T_{i+1}$ is immediately next to $T_i$ in $\T_n$. Since the algorithm starts with the minimum tree and stops with the maximum tree in $\T_n$, the sequence $T_1,T_2,\ldots$ determined by the algorithm contains all the elements of $\T_n$. Therefore, the sequence $T_1^0,T_1^1,T_2^0,T_2^1,\ldots$ generated by the algorithm contains all the cotrees with $n$ leaves.

\medskip

\noindent (b) By item (a), the associated sequence $S = G_1^0,G_1^1,G_2^0,G_2^1,\ldots$ of cographs contains all the members of $\C_n$. Now, note that the $T_i$'s are in increasing order according to Definition~\ref{def:ordering-trees}, and by Lemma~\ref{lem:equality-isomorphism} are pairwise nonisomorphic. Hence, $S$ is in {\em strictly} increasing order according to Definition~\ref{def:cograph-ordering}. This implies that no two cographs in $S$ are isomorphic.

\medskip

\noindent (c) By Fact~\ref{fact:number-of-internal-nodes}, each $T_i$ contains $O(n)$ nodes; hence, obtaining $T_i^0$ and $T_i^1$ from $T_i$ can be done in $O(n)$ time. In addition, by Corollary~\ref{coro:next-tree-complexity}, $\mathit{NextTree}(T_i)$ runs in $O(n)$ time. Therefore, Algorithm~\ref{alg:cograph-generation} has delay $O(n)$.

\medskip

\noindent (d) It is easy to see that $T_{\mathit{min}}$ can be determined in $O(n)$ time. This implies that the time spent to output $T_1^0$ is $O(n)$.
\end{proof}

\begin{coro}\label{coro}
Algorithm~\ref{alg:cograph-generation} determines a linear order on $\C_n$.
\end{coro}

Figure~\ref{fig:four-vertex-cographs} depicts in increasing order all cographs with $4$ vertices generated by $\mathit{CographGeneration}(4)$.

\begin{figure}[H]
\centering
\includegraphics[scale=0.7]{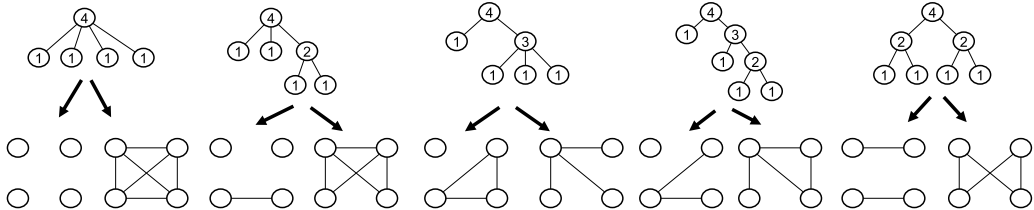}
\caption{Cographs with 4 vertices}
\label{fig:four-vertex-cographs}
\end{figure}

\section{Conclusions}

Based on the fact that a cotree is a compact, $O(n)$ size representation of an $n$-vertex cograph, in this work we described an algorithm for generating all unlabeled cographs with $n$ vertices, via the generation of cotrees. The delay of our algorithm is $O(n)$ and the time needed to generate the first output is also $O(n)$, yielding an overall $O(n\,M_n)$ time complexity, where $M_n=|\C_n|$. The algorithm avoids the generation of duplicates (isomorphic outputs) and produces, as a by-product, a total ordering of $\C_n$. To the best of the authors' knowledge, this is the first practical description of a linear-delay cograph generation algorithm.

The algorithm was implemented in language $C\#$ and executed on an AMD FX-6100 Six-Core Processor at 3.30GHz with 8GB of RAM,  running \textit{Windows 7} operating system. Table~\ref{tab:number-of-cographs} shows the values $|\C_n|$ calculated by the algorithm, for $n\leq 19$. We remark that the results in Table~\ref{tab:number-of-cographs} agree with the results presented in~\cite{hougardy}, for $n\leq 10$. In a future work, we will describe an application of our generation algorithm in the search of counter-examples for spectral graph theory conjectures on cographs.

\begin{table}[H]
\begin{small}
\begin{center}
\begin{tabular}{|c|c|c|c|c|c|c|c|c|c|c|c|c|}
\hline
$n$ & 2 & 3 & 4 & 5 & 6 & 7 & 8 & 9 & 10 & 11 & 12 & 13 \\
\hline
$|\C_n|$ & 2 & 4 & 10 & 24 & 66 & 180 & 522 & 1532 & 4624 &  14136 & 43930 & 137908 \\
\hline
\end{tabular}

\begin{tabular}{|c|c|c|c|c|c|c|}
$n$ & 14 & 15 & 16 & 17 & 18 & 19 \\
\hline
$|\C_n|$  & 437502 & 1399068 & 4507352 & 14611576 & 47633486 & 156047204\\
\hline
\end{tabular}
\end{center}
\end{small}
\caption{Number of cographs with $n\leq 19$.\label{tab:number-of-cographs}}
\end{table}

\bibliographystyle{elsarticle-num}

\end{document}